\documentclass[document, 11]{article}
\usepackage[utf8]{inputenc}
\usepackage[usenames,dvipsnames,svgnames,table]{xcolor}
\usepackage{amsmath, amssymb, graphicx, amsthm, mathtools,etoolbox, mathrsfs, boxedminipage,thmtools,thm-restate}
\usepackage{enumerate}
\usepackage{import}
\usepackage{color}
\usepackage{todonotes}

\usepackage[colorlinks=true,
            linkcolor=NavyBlue,
            urlcolor=blue,
            citecolor=OliveGreen]{hyperref}


\usepackage[scaled]{helvet} 
\usepackage{courier} 
\normalfont
\usepackage[T1]{fontenc}

\usepackage{framed}
\numberwithin{equation}{section}

\theoremstyle{plain}
\newtheorem{theorem}{Theorem}

\newtheorem{lemma}[theorem]{Lemma}

\newtheorem{claim}[theorem]{Claim}

\theoremstyle{definition}

\theoremstyle{enumtheo}

\newtheoremstyle{break}
  {\topsep}
  {\topsep}
  {}
  {}
  {\bfseries}
  {.}
  {\newline}
  {}
\theoremstyle{break}

\newtheorem{algo}[theorem]{Algorithm}

\newenvironment{algorithm}[3]
        {\begin{boxedminipage}{\textwidth}\begin{algo}[#1]
        {\begin{tabular}{r l}
        \textbf{Intput} & #2\\
        \textbf{Output} & #3
        \end{tabular}\par\enskip}}
        {\end{algo}\end{boxedminipage}\newline}

\theoremstyle{remark}
\newtheorem{remark}[theorem]{Remark}

\theoremstyle{corollary}
\newtheorem{corollary}[theorem]{Corollary}

\usepackage{tikz}

\DeclareSymbolFont{extraup}{U}{zavm}{m}{n}
\DeclareMathSymbol{\varheart}{\mathalpha}{extraup}{86}
\DeclareMathSymbol{\vardiamond}{\mathalpha}{extraup}{87}

\def\tif{\text{ if }}
\def\tand{\text{ and }}

\def\RR{\mathbb{R}}

\def\Par{\mathcal{P}}

\def\matr#1{\mathsf{#1}}
\def\vecc#1{\boldsymbol{#1}}

\def\ee{\varepsilon}
\def\eps{\epsilon}
\def\rdim{{\cal R}_{diam}}

\DeclareMathOperator{\Exp}{\mathbb E}
\DeclareMathOperator{\Pp}{\mathbb P}
\DeclareMathOperator{\Ee}{\mathcal E}
\DeclareMathOperator{\ddeg}{deg}
\DeclareMathOperator{\OO}{\mathcal O}
\DeclareMathOperator{\OOs}{\widetilde{\mathcal O}}

\DeclareMathOperator{\plog}{polylog}
\DeclareMathOperator{\vol}{vol}
\DeclareMathOperator{\Diff}{\Delta}

\DeclareMathOperator{\poly}{poly}

\DeclarePairedDelimiter\set{\lbrace}{\rbrace}
\DeclarePairedDelimiter\parens{\lparen}{\rparen}

\DeclarePairedDelimiter\norm{\|}{\|}
\DeclarePairedDelimiter\Abs{|}{|}
\DeclarePairedDelimiter\inpr{\langle}{\rangle}

\DeclareMathOperator{\Reff}{\mathsf{Reff}}

\def\RRn{\RR_{\ge 0}}
\def\2Lin{\mathsf{2Lin}}

\begin{document}
\author{
Vedat Levi Alev\footnote{Supported by the David R.~Cheriton Graduate Scholarship. E-mail: \href{mailto:vlalev@uwaterloo.ca}{vlalev@uwaterloo.ca}}
\and Nima Anari \footnote{ Email: \href{mailto:anari@stanford.edu}{anari@stanford.edu}}
\and Lap Chi Lau\footnote{Supported by NSERC Discovery Grant 2950-120715 and NSERC Accelerator Supplement 2950-120719. E-mail: \href{mailto:lapchi@uwaterloo.ca}{lapchi@uwaterloo.ca}}
\and Shayan Oveis Gharan\footnote{Supported by NSF Career Award and  ONR Young Investigator Award, Email: \href{mailto:shayan@cs.washington.edu}{shayan@cs.washington.edu}
}
}
\title{Graph Clustering using Effective Resistance}
\date{}
\maketitle
\begin{abstract}
  We design a polynomial time algorithm that for any weighted undirected graph $G = (V, E,\vecc w)$  and sufficiently large $\delta > 1$,
  partitions $V$ into subsets
  $V_1, \ldots, V_h$ for some $h\geq 1$, such that
  \begin{itemize}
  \item at most $\delta^{-1}$ fraction of the weights are between clusters, i.e.
    \[ w(E - \cup_{i = 1}^h E(V_i)) \lesssim \frac{w(E)}{\delta};\]
  \item the effective resistance diameter of each of the induced
    subgraphs $G[V_i]$ is at most $\delta^3$ times the average weighted degree, i.e.
    \[ \max_{u, v \in V_i} \Reff_{G[V_i]}(u, v) \lesssim
    \delta^3 \cdot \frac{|V|}{w(E)} \quad \text{ for all } i=1, \ldots, h.\]
  \end{itemize}
  In particular, it is possible to remove one percent of weight of
  edges of any given graph such that each of the resulting connected
  components has effective resistance diameter at most the inverse of
  the average weighted degree.

  Our proof is based on a new connection between effective resistance and low conductance sets.
  We show that if the effective resistance between two vertices $u$ and $v$ is large, then there must be a low conductance cut separating $u$ from $v$.
  This implies that very mildly expanding graphs have constant effective resistance diameter.
  We believe that this connection could be of independent interest in algorithm design.
\end{abstract}
\section{Introduction}

Graph decomposition is a useful algorithmic primitive with various applications.
The general framework is to remove few edges so that the remaining components have nice properties, and then  specific problems are solved independently in each component.
Several types of graph decomposition results have been studied in the literature.
The most relevant to this work are low diameter graph decompositions and expander decompositions.
We refer the reader to Section~\ref{sec:prelim} for notation and definitions.

\paragraph{Low Diameter Graph Decompositions:}
Given a weighted undirected graph $G = (V, E, \vecc w)$ and a parameter
$\Delta > 0$, a low diameter graph decomposition algorithm seeks to
partition the vertex set $V$ into sets $V_1, \ldots, V_h$ with the following two properties:
\begin{itemize}
\item Each component $G[V_i]$ has bounded shortest path diameter, i.e.~$\max_{u, v \in V_i} {\rm dist}_w(u, v) \le \Delta$,
where ${\rm dist}_w(u,v)$ is the shortest path distance between $u$ and $v$ using the edge weight $\vecc w$.
\item There are not too many edges between the sets $V_i$, i.e.~$\Abs*{E - \bigcup_{i = 1}^h E(V_i) } \le \frac{D(G)}{\Delta} \cdot |E|$,
where $D(G)$ is the ``distortion'' that depends on the input graph.
\end{itemize}
This widely studied \cite{LinialSaks, KleinPlotkinRao, Bartal1,
  LeeSid, CopsRobbers} primitive (and its generalization to
decomposition into padded partitions) has been very useful in
designing approximation algorithms~\cite{DirectedSteiner, 0ext1,
  0ext2, FHL08, RoughgardenKrauthgamer, BansalFeigeKrauthgamer,LOT14}.  This
approach is particularly effective when the input graph is of bounded
genus $g$ or $K_r$-minor free, in which case
$D(G) = \OO(\log g)$~\cite{LeeSid} and
$D(G) = \OO(r)$~\cite{CopsRobbers}.  For these special graphs, this
primitive can be used to proving constant flow-cut
gaps~\cite{KleinPlotkinRao}, proving tight bounds on the Laplacian
spectrum~\cite{BiswalLeeRao, HigherEigenvalues}, and obtaining
constant factor approximation algorithms for NP-hard
problems~\cite{BansalFeigeKrauthgamer,AlevLau}.  However, there are
graphs for which $D(G)$ is necessarily $\Omega(\log n)$
where $n$ is the number of vertices, and this
translates to a $\Omega(\log n)$ factor loss in applying this approach
to general graphs.  For example, in a hypercube, if we only delete a
small constant fraction of edges, some remaining components will have
diameter $\Omega(\log n)$.

\paragraph{Expander Decompositions:}
Given an undirected graph $G = (V, E)$ and a parameter $\phi > 0$,
an expander decomposition algorithm seeks to partition the vertex set $V$
into sets $V_1, \ldots, V_h$ with the following two properties.
\begin{itemize}
\item Each component $G[V_i]$ is a $\phi$-expander, i.e.~$\Phi(G[V_i]) \ge \phi$, where $\Phi(G[V_i])$ is the conductance of the induced subgraph $G[V_i]$; see Section~\ref{sec:prelim} for the definition of conductance.
\item There are not too many edges between the sets $V_i$,
  i.e.~$\Abs*{E - \bigcup_{i = 1}^h E(V_i) } \le \delta(G, \phi) \cdot |E|$, where $\delta(G,\phi)$ is a parameter depending on the graph $G$ and $\phi$.
\end{itemize}
This decomposition is also well studied~\cite{KannanVempalaVetta, SpecSpars, AroraBarakSteurer, ShayanLuca}, and is proved useful in solving Laplacian equations, approximating Unique Games, and designing clustering algorithms.
It is of natural interest to minimize the parameter $\delta(G, \phi)$.
Similar to the low diameter partitioning case,
there are graphs where $\delta(G, \phi) \ge \Omega(\phi \cdot \log(n) )$.
For example, in a hypercube, if we  delete a small constant fraction of edges, some remaining components will have conductance $\OO(1/\log n)$.

\paragraph{Motivations:} In some applications, we could not afford to have
an $\Omega(\log n)$ factor loss in the approximation ratio.  One
motivating example is the Unique Games problem.  It is known that
Unique Games can be solved effectively in graphs with constant
conductance~\cite{UGexpander} and more generally in graphs with low
threshold rank~\cite{Kolla, GuruswamiSinop,BarakRaghavendraSteurer},
and in graphs with constant diameter~\cite{GuptaTalwar}.  Some
algorithms for Unique Games on general graphs are based on graph
decomposition results that remove a small constant fraction of edges
so that the remaining components are of low threshold
rank~\cite{AroraBarakSteurer} or of low diameter~\cite{AlevLau}, but
the $\Omega(\log n)$ factor loss in the decomposition is the
bottleneck of these algorithms.  This leads us to the question of
finding a property that is closely related to low diameter and high
expansion, so that every graph admits a decomposition into components
with such a property without an $\Omega(\log n)$ factor loss.

\paragraph{Effective Resistance Diameter:}
The property that we consider in this paper is having low effective
resistance diameter.  We interpret the graph $G = (V, E, \vecc w)$ as an
electrical circuit by viewing every edge $e \in E$ as a resistor with
resistance $1/w(e)$. The effective resistance distance $\Reff(u, v)$
between the vertices $u$ and $v$ is then the potential difference
between $u$ and $v$ when injecting a unit of electric flow into the
circuit from the vertex $u$ and removing it out of the circuit from
the vertex $v$.  We define
\[ \rdim(G) :=\max_{u,v \in V} \Reff(u,v)\]
as the effective resistance diameter of $G$.  Both the properties of
low diameter and of high expansion have the property of low effective
resistance diameter as a common denominator:
The effective resistance distance $\Reff(u, v)$ is
upper bounded by the shortest path distance for any graph, and so every
low diameter component has low effective resistance diameter.
Also, a $d$-regular graph with constant expansion has effective resistance diameter $\OO(1/d)$ \cite{BK89,ChandraEtAl}, and so an expander graph also has low effective resistance diameter.
See~Section~\ref{sec:prelim} for more details.

In this paper, we study the connection between effective resistance
and graph conductance.  Roughly speaking, we show if all sets have
mild expansion (see \autoref{thm:wtshnim}), then the effective
resistance diameter is small.  We use this observation to design a
graph partitioning algorithm to decompose a graph into clusters with
effective resistance diameter at most the inverse of the average
degree (up to constant losses) while removing only a constant fraction
of edges.  This shows that although we cannot partition a graph into
$\Omega(1)$-expanders by removing a constant fraction of edges, we can
partition it into components that satisfy the ``electrical
properties'' of expanders.

\paragraph{Applications of Effective Resistance:} Besides the motivation
from the Unique Games problem, we believe that effective resistance is
a natural property to be investigated on its own.  The effective
resistance distance between two vertices $u, v \in V$ has many useful
probabilistic interpretations, such as the commute
time~\cite{ChandraEtAl}, the cover time~\cite{Matthews}, and the
probability of an edge being in a random spanning
tree~\cite{Kirchoff}.  See Section~\ref{sec:prelim} for more details.
Recently, the concept of effective resistance has found surprising
applications in spectral sparsification~\cite{SpielmanSrivastava}, in
computing maximum flows~\cite{CKMST}, in finding thin
trees~\cite{NimaShayan}, and in generating random spanning
trees~\cite{MadryKelner,MadryStraszakTarnawski,DurfeeEtAl}.  The
recent algorithms in generating a random spanning tree are closely
related to our work.  Madry and Kelner~\cite{MadryKelner} showed how
to sample a random spanning tree in time $\OOs(m \cdot \sqrt n)$
where $m$ is the number of edges,
faster than the worst case cover time $\OOs(m \cdot n)$
(see Section~\ref{sec:prelim}). A cruicial ingredient of
their algorithm is the low diameter graph decomposition technique,
which they use to ensure that the resulting components have small
cover time.  In subsequent work, Madry, Straszak and Tarnawski
\cite{MadryStraszakTarnawski} have improved the time complexity of
their algorithm to $\OOs(m^{4/3})$ by working with the effective
resistance metric instead of the shortest path metric.  Indeed, their
technique of reducing the effective resistance diameter is similar to
our technique -- even though it cannot recover our result.

\subsection{Our Results}

Our main technical result is the following connection between effective resistance and graph partitioning.

\begin{restatable}{theorem}{wtshnim}\label{thm:wtshnim}
  Let $G = (V, E)$ be a weighted graph with weights
  $\vecc w \in \RRn^E$. Suppose for any set $S \subseteq V$ with
  $\vol(S) \le \vol(G)/2$ we have
  \begin{equation}
    \Phi(S) \ge \frac{c}{\vol(S)^{1/2 - \ee}} \label{eq:wbdbd}\tag{mild expansion}
  \end{equation}
  for some $c > 0$ and $1/2 \ge \ee \ge 0$.
Then, for any pair of vertices $s, t \in V$, we have
  \begin{equation}
    \Reff(s, t) \lesssim \parens*{\frac{1}{\deg(s)^{2\ee}} + \frac{1}{\deg(t)^{2\ee}  }} \cdot \frac{1}{\ee \cdot c^2}, \tag{resistance bound}
  \end{equation}
where $\deg(v)=\sum_{u: uv\in E} w(u,v)$ is the weighted degree of $v$.
\end{restatable}

In~\cite{ChandraEtAl}, Chandra et al. proved that a $d$-regular graph
with constant expansion has effective resistance diameter $\OO(1/d)$.
They also proved that the effective resistance diameter of a
$d$-dimensional grid is $\OO(1/d)$ when $d>2$ even though it is a poor
expander.  Theorem~\ref{thm:wtshnim} can be seen as a common
generalization of these two results, using the mild expansion
condition as a unifying assumption.  Chandra et al.~\cite{ChandraEtAl}
also showed that the effective resistance diameter of a
$2$-dimensional grid is $\Theta(\log n)$.  Note that for a
$\sqrt{n}\times \sqrt{n}$ grid, $\Phi(S)\approx 1/\vol(S)^{1/2}$ for
any $k\times k$ square.  This shows that the mild expansion assumption
of the theorem cannot be weakened in the sense that if $\ee=0$ for
some sets $S$, then $\Reff(s,t)$ may grow as a function of $|V|$.

The proof of Theorem~\ref{thm:wtshnim} also provides an efficient
algorithm to find such a sparse cut.  The high-level idea is to prove that if all
level sets of the $st$ electric potential vector satisfy the
\ref{eq:wbdbd} condition, then the potential difference between $s$
and $t$ must be small, i.e., $\Reff(s,t)$ is small.  Combining with a fast Laplacian solver \cite{SpielmanTengLapSolv}, we
show that the existence of a pair of vertices $u, v \in V$ with high
effective resistance distance implies the existence of a sparse cut which can be found in nearly linear time.

\begin{restatable}{corollary}{smallcut}
\label{cor:smallcut}
  Let $G = (V, E, \vecc w)$ be a weighted undirected graph. If $\deg(v)\geq
  1/\alpha$ for all $v\in V$, then for any $0<\ee<1/2$, there is a
  subset of vertices $U \subseteq V$ such that
\[\Phi(U)  \lesssim \frac{\alpha^\ee}{\sqrt{\rdim\cdot \ee}} \cdot \vol(U)^{\ee - 1/2}.\]
  Furthermore, the set $U$ can be found in time $\OOs\parens*{m \cdot
    \log\parens*{\frac{w(E)}{\min_e w(e)}}}$.
\end{restatable}

Using Corollary~\ref{cor:smallcut} repeatedly, we can prove the following graph decomposition result.

\begin{restatable}[Main]{theorem} {lapsolv}
\label{thm:lapsolv}
  Given a weighted undirected graph $G = (V, E, \vecc w)$, and a large enough
  parameter $\delta > 1$, there is an algorithm with time complexity $\OOs\parens*{m \cdot n \cdot \log\parens*{\frac{w(E)}{\min_e w(e)}}}$ that finds a partition $V = \bigcup_{i = 1}^h V_i$ satisfying
  \begin{equation}
    w\parens*{ E - \bigcup_{i = 1}^h E(V_i) } \lesssim   \frac{w(E)}{\delta} \tag{loss bound} \label{eq:effresbd2}
  \end{equation}
  and
  \begin{equation}
    \rdim(G[V_i]) \lesssim \delta^{3} \cdot \frac{n}{w(E)} \tag{resistance bound} \label{eq:resbd}
  \end{equation}
  for all $i=1, \ldots,h$.
\end{restatable}

Let $G$ be a $d$-regular unweighted graph.  Theorem~\ref{thm:lapsolv}
implies that it is possible to remove a constant fraction of the edges
of $G$ and decompose $G$ into components with effective resistance
diameter at most $\OO(1/d)$. Note that $d$-regular
$\Omega(1)$-expanders with $\rdim=\OO(1/d)$ have the least effective
resistance diameter among all $d$-regular graphs. So, even though it
is impossible to decompose $d$-regular graphs into graphs with
$\Omega(1)$-expansion while removing only a constant fraction of
edges, we can find a decomposition with analogous ``electrical
properties''.

We can also view Theorem~\ref{thm:lapsolv} as a
generalization of the following result: Any $d$-regular graph can be
decomposed into $\Omega(d)$-edge connected subgraphs by removing only
a constant fraction of edges.  This is because if the effective
resistance diameter of an unweighted graph $G$ is $\eps$, then $G$
must be $1/\eps$-edge connected. Recall that a graph is $k$-edge
connected, if the size of every cut in that graph is at least $k$.

\section{Preliminaries}\label{sec:prelim}

In this section, we will first define the notations used in this paper,
and then we will review the background in effective resistances, Laplacian solvers, and graph expansions in the following subsections.

Given an undirected graph $G = (V, E)$ and a subset of vertices $U \subseteq V$, we use the notation $E_G(U)$ for the set of edges with both endpoints in $U$,
i.e.~$E_G(U) = \set{\set{u, v} \in E(G) : u, v \in U}$.
We write $U^c$ for the complement of $U$ with respect to $V(G)$, i.e.~$U^c = V\backslash U$.
The variables $n$ and $m$ stand for the number of vertices and the edges of the graph respectively, i.e.~$n = |V|$ and $m = |E|$.
We use the notation $\partial_G U$ for the edge boundary of $U \subseteq V$,
i.e.~$\partial_G U = E_G(U, U^c) = \set{\set{u, v} \in E : u \in U, v \in
  U^c}$.
For a graph $G = (V, E)$ with weights $\vecc w \in \RRn^E$, we
write $\deg_G(v) = \sum_{u : uv \in E} w(u, v)$ for the weighted degree
of $v$. For $S \subseteq V$, the volume $\vol_G(S)$ of $S$ is defined to
be $\vol_G(S) = \sum_{s \in S} \ddeg(s)$. When the graph is clear in the context we may drop the subscript in all aforementioned notation.

Scalar functions and vectors are typed in bold,
i.e.~$\vecc x \in \RR^V$, or $\vecc w \in \RR^E$. For a subset
$A \subseteq E$, the notation $w(A)$ stands for the sum of the weights
of all edges in $A$, i.e.~$w(A) = \sum_{e \in A} w(e)$. The $j$-th
canonical basis vector is denoted by $\vecc e_j \in \RR^V$.  Matrices
are typed in serif, i.e.~$\matr A \in \RR^{V \times V}$.

Time complexities are given in asymptotic notation. We employ the
notation $\OOs(f(n))$ to hide polylogarithmic factors in $n$,
i.e.~$\OOs(f(n)) = \OO(f(n) \cdot \plog(n))$.  We use the notation $f
\lesssim g$ for asymptotic inequalities, i.e.~$f = \OO(g)$; and the
notation $f \asymp g$ for asymptotic equalities, i.e.~$f = \Theta(g)$.

\subsection{Electric Flow, Electric Potential, and Effective Resistance}

Let $G = (V, E)$ be a given graph with non-negative edge weights
$\vecc w \in \RRn^E$. The notion of an electric flow arises when one
interprets the graph $G$ as an electrical network where every edge $e
\in E$ represents a resistor with resistance $1/w(e)$.

We fix an arbitrary orientation $E^\pm$ of the edges $E$ and define a
unit $st$ flow in this network as a function $\vecc f \in
\RRn^{E^\pm}$ (where for $e \not\in E^{\pm}$ we define $f(e) =
-f(-e)$) satisfying the following:
\[ \sum_{w \in \delta^+(v)} f(vw) - \sum_{u \in \delta^-(v)} f(uv) =
  \begin{cases}
    1 & \tif v = s\\
    -1 & \tif v = t\\
    0 & \text{ otherwise, } \tag{flow conservation}\label{eq:fc}
  \end{cases}\]
where $\delta^+(v)$ is the set of edges having $v$ as the
head in our orientation, and $\delta^-(v)$ is the set of
edges having $v$ as tail.  Let $e = uv \in E^\pm$ be an oriented edge.
The flow $\vecc f$ has to obey Ohm's law
\begin{equation}
  f(e) = w(e) \cdot \Diff_e\vecc p = w(e) \cdot (p(u) - p(v)) \tag{Ohm's law}\label{eq:ohm}
\end{equation}
for some vector $\vecc p \in \RR^V$ which we call the potential
vector.
The electrical flow between the vertices $s$ and $t$ is the unit $st$ flow that satisfies \ref{eq:fc} and \ref{eq:ohm}.

The electrical energy $\Ee(\vecc f)$ of a flow $\vecc f$ is defined as
the following quantity,
\begin{equation}
\Ee(\vecc f) = \sum_{e \in E^\pm} \frac{f(e)^2}{w(e)}.  \tag{electrical energy}
\end{equation}
It is known that the electric flow between $s$ and $t$ is the unit
$st$ flow with minimal electrical energy.  The effective resistance
$\Reff(s, t)$ between the vertices $s$ and $t$ is the potential
difference between the vertices $s$ and $t$ induced by this flow,
i.e. $\Reff(s, t) = p(s) - p(t)$. It is known that the potential
difference between $s$ and $t$ equals the energy $\Ee(\vecc f_{st})$
of this flow. This is often referred as Thomson's principle.

The electric potential vector and the effective resistance are known
to have the following closed form expressions: Let $\matr W \in \RR^{V
  \times V}$ be the weighted adjacency matrix of $G$, i.e.~the matrix
satisfying $W(u, v) = 1[uv \in E] \cdot w(u, v)$, and $\matr D \in
\RR^{V \times V}$ the weighted degree matrix, i.e.~the diagonal matrix
satisfying $D(v, v) = \ddeg(v) = \sum_{u : uv \in E} w(u, v)$.  The
(weighted) Laplacian $\matr L_G \in \RR^{V \times V}$ is defined to be
the matrix
\begin{equation}
  \matr L_G = \matr D - \matr W. \tag{weighted Laplacian}
\end{equation}
It is well-known that this is a symmetric positive semi-definite
matrix. We will take
\[ \matr L_G = \sum_{i = 2}^n \lambda_i \vecc v_i \vecc v_i^\top \]
as the spectral decomposition of $\matr L_G$, where
$\lambda_1 = 0 \le \lambda_2 \le \cdots \le \lambda_n$ are the
eigenvalues of $\matr L_G$ sorted in increasing order. It is easy to
verify $\matr L_G \vecc 1 = 0$ and further it can be shown that this
is the only vector (up to scaling) satisfying this when $G$ is
connected. This means if $G$ is connected, the matrix $\matr L_G$ is
invertible in the subspace perpendicular to $\vecc 1$. This inversion
will be done by the matrix $\matr L_G^\dagger$, the so-called
Moore-Penrose pseudo-inverse of $\matr L_G$ defined by
\begin{equation}
  \matr L_G^\dagger = \sum_{j = 2}^n \frac{1}{\lambda_j} \vecc v_i
  \vecc v_i^\top. \tag{pseudo-inverse of $\matr L_G$}
\end{equation}

Let $\vecc f^\star \in \RR^E$ be the $st$ unit electric flow vector. It can
be verified that the $st$ electric potential $\vecc p^\star$ --
i.e.~the vector satisfying $w(uv) \cdot (p^\star(u) - p^\star(v)) =
f^\star(uv)$ for all $uv \in E^\pm$ -- satisfies the equation
\begin{equation}
\matr L_G \vecc p^\star = \vecc e_s - \vecc e_t  \Longleftrightarrow  \vecc p^\star = \matr L_G^\dagger (\vecc e_s - \vecc e_t). \label{eq:lapp}
\end{equation}

In particular, this implies the following closed form expression for
$\Reff(s, t)$
\begin{equation}
  \Reff(s, t) = \langle \vecc e_s - \vecc e_t, \matr L_G^\dagger (\vecc e_s - \vecc e_t) \rangle. \tag{$st$ effective resistance} \label{eq:effresdef}
\end{equation}
It can be verified that this defines a ($\ell_2^2$) metric on the set
vertices $V$ of $G$ \cite{KleinRandic}, as we have
\begin{enumerate}
\item $\Reff(u, v) = 0$ if and only if $u = v$.
\item $\Reff(u, v) = \Reff(v, u)$ for all $u, v \in V$.
\item $\Reff(u, v) + \Reff(v, w) \ge \Reff(u, w)$ for all $u, v, w \in V$.
\end{enumerate}
Further, by routing the unit $st$ flow along the $st$ shortest path we
see that the shortest path metric dominates the effective resistance
metric, i.e.  $\Reff(u, v) \le \text{dist}(u, v)$ for all the pairs of vertices
$u, v \in V$.

It is known that the commute time distance $\kappa(u, v)$ between $u$
and $v$ -- the expected number of steps a random walk starting from
the vertex $u$ needs to visit the vertex $v$ and then return to $u$ -- is
$\vol(G)$ times the effective resistance distance $\Reff(u, v)$
\cite{ChandraEtAl}.  Also, the effective resistance $\Reff(u,v)$ of an
edge $uv \in E$ corresponds to the probability of this edge being
contained in a uniformly sampled random spanning tree \cite{Kirchoff}.
A well-known result of Matthews \cite{Matthews} relates the effective
resistance diameter to the cover time of the graph -- the expected
number of steps a random walk needs to visit all the vertices of
$G$. Aldous \cite{AldousRST} and Broder \cite{BroderRST} have shown
that simulating a random walk until every vertex has been visited
allows one to sample a uniformly random spanning tree of the graph.

\subsection{Solving Laplacian Systems}\label{sec:lapsolv}

For our algorithmic results, it will be important to be able to
compute electric potentials, and effective resistances quickly. We
will do this by appealing to Equation \eqref{eq:lapp} and the
definition of the \ref{eq:effresdef}. Both of these equations require
us to solve a Laplacian system. Fortunately, it is known that these
systems can be solved in nearly linear time \cite{SpielmanTengLapSolv,
  KoutisMillerPeng10, KoutisMillerPeng11, KelnerEtAl, cohen2014solving, KyngSachdeva}.

\begin{lemma}[The Spielman-Teng Solver, \cite{SpielmanTengLapSolv}]\label{lem:stsolver}
  Let a (weighted) Laplacian matrix $\matr L \in \RR^{V \times V}$, a
  right-hand side vector $\vecc b \in \RR^V$, and an accuracy parameter
  $\zeta > 0$ be given.
  Then, there is a randomized algorithm which takes time
  $\OOs(m \cdot \log(1/\zeta))$ and produces a vector $\widehat{\vecc x}$
  that satisfies
  \begin{equation}
    \norm*{\widehat{\vecc x} - \matr L^\dagger \vecc b}_{\matr L} \le \zeta \cdot \norm{\matr L^\dagger \vecc b}_{\matr L} \tag{accuracy guarantee}\label{eq:ag}
  \end{equation}
  with constant probability, where $\norm{\vecc x}^2_{\matr A} = \langle \vecc x, \matr A \vecc x \rangle$.
\end{lemma}

For our purposes it will suffice to pick $\zeta$ inversely
polynomial in the size of the graph in the unweighted case, and
$1/ \poly(w(E)/\min_e w(e), 1/m)$ in the weighted case.

Extending the ideas of Kyng and Sachdeva \cite{KyngSachdeva}, Durfee et
al.~\cite{DurfeeEtAl} show that it is possible to compute approximations for effective resistances between a set of given pairs $S \subseteq V \times V$
efficiently.
\begin{lemma}\label{lem:singsourceeffres}
  Let $G = (V, E, \vecc w)$ be a weighted graph, $\beta > 0$ an
  accuracy parameter, and $S \subseteq V \times V$. There is an
  $\OOs(m + (n + |S|)/\beta^2)$-time algorithm which returns numbers
  $A(u, v)$ for all $(u,v) \in V$ satisfying
  \[ e^{-\beta} \Reff(u, v) \le A(u, v) \le e^{\beta} \Reff(u, v). \]
\end{lemma}
This lemma will aid us in computing fast approximations for furthest
points in the effective resistance metric. For our purposes, we only
need to pick $\beta$ as a small enough constant, i.e.~$\beta=\ln(3/2)$.
Similar guarantees can also be obtained using the ideas of Spielman
and Srivastava \cite{SpielmanSrivastava}.

\subsection{Conductance}

For a graph $G = (V, E)$ with non-negative edge weights $\vecc w \in \RRn^E$, we define the conductance of a set $S \subseteq V$ as
\begin{equation}
  \Phi(S) = \frac{w(\partial S)}{\vol(S)} \tag{conductance of a set}.
\end{equation}
The conductance of the graph $G$ is then defined as
\begin{equation}
  \Phi(G) = \min\set*{\Phi(S) : S \subseteq V \tand 2 \vol(S) \le \vol(G)}.
  \tag{conductance of a graph}
\end{equation}
It is well-known \cite{Cheeger, AlonMilman} that the conductance of
the graph $G$ is controlled by the spectral gap (second smallest
eigenvalue) $\tilde \lambda_2$ of the normalised Laplacian matrix
$\matr D^{-1/2} \matr L_G \matr D^{-1/2}$, i.e.
\begin{equation}
 \tilde \lambda_2 \lesssim \Phi(G) \lesssim \sqrt{\tilde
    \lambda_2}. \tag{Cheeger's inequality}\label{eq:cheeg}
\end{equation}
Appealing to the closed form formula for the \ref{eq:effresdef} it can
be verified that the spectral gap $\lambda_2$ of the (unnormalised) Laplacian
controls the effective resistance distance, i.e.
\[ \max_{s, t \in V} \Reff(s, t) \lesssim \frac{1}{\lambda_2}.\]
By an easy application of \ref{eq:cheeg} we see that the expansion
controls the effective resistance as well, i.e.
\[  \max_{s, t \in V} \Reff(s, t) \lesssim \frac{1}{\Phi(G)^2}.\]
Indeed,  Theorem \ref{thm:wtshnim} and Corollary \ref{cor:smallcut} will improve upon this bound.

\section{From Well Separated Points to Sparse Cuts} \label{sec:shnimproof}

In this section, we are going to prove Theorem~\ref{thm:wtshnim} and Corollary~\ref{cor:smallcut}.
As previously mentioned, we will prove that if all the
level sets of the potential vector have \ref{eq:wbdbd},
the effective resistance cannot be high.

\wtshnim*
\begin{proof}
In the following let $\vecc f \in \RR^E$ be a unit electric flow from
$s$ to $t$, and $\vecc p \in \RR^V$ be the corresponding vector of
potentials where we assume without loss of generality that $p(t) = 0$.
We direct our attention to the following threshold sets
\[  S_p = \set{v \in V: \vecc p(v) \ge p}.\]
Then, we have
\[  \sum_{e \in \partial S_p} | f(e)| = 1.\]
Using \ref{eq:ohm}, we can rewrite this into
\begin{equation}
  \sum_{e \in \partial S_p} w(e) \cdot \Abs{\Diff_e \vecc p} = 1,\label{eq:conserv1}
\end{equation}

where $\Diff_e \vecc p$ is the potential difference along the
endpoints of the edge $e$.
Normalizing this, we get
\begin{equation}
  \sum_{e \in \partial S_p} \frac{w(e)}{w(\partial S_p)} \cdot \Abs{\Diff_e \vecc p} = \frac{1}{w(\partial S_p)}. \label{eq:dropsum}
\end{equation}
Now, set $\mu(e) = w(e) / w(\partial S_p)$. Restricted over
the set of edges $\partial S_p$, $\mu$ is a probability
distribution and the LHS of \eqref{eq:dropsum} corresponds to
the expected potential drop when edges $e \in \partial S_p$
are sampled with respect to the probability distribution $\mu$, i.e.~we
have
\[  \Exp_\mu \Abs{\Diff_e \vecc p} = \frac{1}{w(\partial S_p)}.\]
Then, by Markov's inequality, we get a set $F \subseteq \partial S_p$ such that
\begin{itemize}
\item all edges $f \in F$ satisfy
  \[ \Abs{\Diff_f \vecc p} \le \frac{2}{w(\partial S_p)}; \]
\item $\Pp_\mu(e \in F) \ge 1/2$, equivalently
    \[ w(F) = \sum_{e \in F} w(e) = \sum_{e \in F} w(\partial S_p) \cdot \mu(e) = w(\partial S_p) \cdot \mu(F) \ge \frac{w(\partial S_p)}{2}. \]
\end{itemize}
Using the observation that the endpoint of an edge $f \in F$ that is
not contained in $S_p$ should have potential at least $p -
2/w(\partial S_p)$, we obtain
\[  \vol(S_{p - 2/w(\partial S_p)})\geq \vol(S_p) + w(F)  \ge \vol(S_p) + \frac{w(\partial S_p)}{2}.\]
Assuming $\vol(\partial S_p)\leq \vol(G)/2$, using the \ref{eq:wbdbd}
property, we have $w(\partial S_p) \geq c\vol(S_p)^{1/2+\ee}$.  So,
from above we get
$$ \vol(S_{p-2/c\vol(S_p)^{1/2+\ee}}) \geq \vol(S_{p-2/w(\partial
  S_p)}) \geq \vol(S_p) + \frac{c\vol(S_p)^{1/2+\ee}}{2},$$ where
in the first inequality we also used that $\vol(S_p)$ increases as $p$
decreases.  Now, iterating this procedure $2
\vol(S_p)^{1/2-\ee}/c$ times we obtain
\begin{equation}
\vol\parens*{ S_{p - \frac{4}{c^2\vol(S_p)^{2\ee}}}} =
\vol\parens*{S_{p - \frac{2}{c\vol(S_p)^{1/2+\ee}} \cdot \frac{2 \vol(S_p)^{1/2-\ee}}{c}}} \ge  2 \vol(S_p),
  \label{eq:voldrop}
\end{equation}
as $\vol(S_p)$ increases as $p$ decreases.  We set $p_0 = p(s)$,
then $\vol(S_{p_0}) = \deg(s)$. Inductively define
\[  p_{k + 1} = p_k - \frac{4}{c^2 \vol(S_{p_k})^{2\ee}}.\]
Then, using the inequality \eqref{eq:voldrop}, we have
\begin{equation}
  \vol(S_{p_{k + 1}}) \ge 2 \cdot \vol(S_{p_k}) \label{eq:volhalf}.
\end{equation}
Note that we can run the above procedure as long as $\vol(S_p) \leq
\vol(G)/2$. Therefore, for some $k^\star \lesssim \log
\frac{\vol(G)}{\deg(s)}$, we must have
\[  \vol(G) \ge 2 \cdot \vol(S_{p_{k^\star}}) \ge \vol(G)/2.\]
Therefore,
\[ p_0 - p_{k^\star} \le 4 \cdot \sum_{j = 0}^{k^\star} \frac{1}{c^2
    \vol(S_{p_i})^{2\ee} }. \]

Using \eqref{eq:volhalf} we get
\[  p_0 - p_{k^\star} \lesssim \frac{1}{c^2 \vol(S_0)^{2 \ee}} \cdot \sum_{j = 0}^{k^\star} \frac{1}{2^{2 j\ee}} \lesssim \frac{1}{\deg(s)^{2\ee} \cdot c^2 \cdot \ee},\]
where the last inequality is a geometric sum with ratio $\approx 1/(1+\eps)$.

By a similar argument (sending flow from $t$ to $s$), we see that more
than half of the vertices have potential smaller than
\[ \frac{1}{\deg(t)^{2\ee}} \cdot \frac{1}{\ee \cdot c^2}.\]
Combining these two bounds, we obtain
\[  \Reff(s, t) =  p(s) \lesssim \parens*{\frac{1}{\deg(s)^{2\ee}} + \frac{1}{\deg(t)^{2\ee}}} \cdot \frac{1}{\ee \cdot c^2}, \]
where the equality follows since the flow is a unit flow.
\end{proof}

\begin{remark}
  For our proof to go through, we do not need the \ref{eq:wbdbd}
  condition to be satisfied by all cuts. It suffices to have this
  condition satisfied by electric potential threshold cuts $(S_p,
  S_p^c)$ only.

  For computational purposes, it will be important to show that our
  argument is robust to small perturbations in the potentials, i.e.~we
  need to show that the proof will still go through when we are
  working with threshold cuts with respect to a vector $\hat{\ p}$
  which is close to the electric potential vector $\vecc p$, rather
  than working with the potential vector $\vecc p$ directly. We will
  show this in Appendix \ref{app:accuracy}, Theorem
  \ref{thm:actualwtshnim}.
\end{remark}

\subsection{Finding the Sparse Cuts Algorithmically}

Next we prove Corollary \ref{cor:smallcut}.

\smallcut*
\begin{proof}
First, we prove the existence of $U$.  Let $u,v\in V$ such that
\begin{equation}\label{eq:maxreffuv}
\Reff(u,v)=\rdim.
\end{equation}
The choice of
\begin{equation} \label{e:c}
c\asymp \sqrt{\frac{\frac1{\deg(u)^{2\ee}}+ \frac1{\deg(v)^{2\ee}}}{\Reff(u,v)\cdot \ee}}
\quad \textrm{ ensures } \quad
\Reff(u,v) > \parens*{\frac{1}{\deg(s)^{2\ee}} + \frac{1}{\deg(t)^{2\ee}  }} \cdot \frac{1}{\ee \cdot c^2}.
\end{equation}
Then, by Theorem~\ref{thm:wtshnim},
there must be a threshold set $S_p$ of the potential vector $\vecc p$ corresponding to sending one unit of electrical flow from $u$ to $v$ such that
\[\Phi(U) \lesssim \frac{c}{\vol(U)^{1/2 - \ee}}
\lesssim \frac{\alpha^\ee}{\sqrt{\ee\cdot \rdim}} \cdot \vol(U)^{\ee - 1/2},
\]
where the last inequality follows from our assumption that $\deg(v) \geq 1/\alpha$ for all $v \in V$.
This proves the first part of the corollary.

It remains to devise a near linear time algorithm to find the set
$U$. First, suppose that we are given the optimum pair of vertices
$u,v$ satisfying \eqref{eq:maxreffuv}.  Using the Spielman-Teng solver
(Lemma \ref{lem:stsolver}), we can compute the potential vector
$\vecc p$ corresponding to sending one unit of electrical flow from to
$u$ to $v$ in time
$\OOs\parens*{m \cdot \log\parens*{\frac{w(E)}{\min_e w(e)}}}$. We can
then sort the vertices by their potential values in time
$\OO(n \log n) = \OOs(m)$.  Finally, we simply go over the sorted list
and find the least expanding level set. This can be done in $\OO(m)$
time in total, since getting $\partial S_{p(v_i)}$ from
$\partial S_{ p(v_{i + 1})}$ (resp.~$\vol(S_{p(v_i)})$ from
$\vol(S_{p(v_{i + 1})})$) can be done by considering the $\deg(v_i)$
edges $e \in \partial(v_i)$ incident to $v_i$.

It remains to find such an optimal pair of vertices $u,v$ satisfying
\eqref{eq:maxreffuv}.  Instead, we find a pair of vertices $u',v'$
such that $\Reff(u',v')\geq \rdim/3$, which is enough for our
purposes as this only causes a constant factor loss in the
conductance of $U$.

\begin{lemma}\label{lem:furthpoints}
  Let $G$ be a weighted graph. In time $\OOs(m)$, one can compute a
  pair of vertices $u, v \in V$ satisfying
  \[ \Reff(u,v)\geq \rdim/3. \]
\end{lemma}
\begin{proof}
  By the triangle inequality for effective resistances, we have the
  following inequality for any $u \in V$:
  \begin{equation}
    \rdim \le 2\max_{v \in V} \Reff(u, v) \label{eq:fixv}.
  \end{equation}
  Thus, we fix a $u \in V$. Appyling Lemma \ref{lem:singsourceeffres}
  (with $S = \set{u} \times V$), we get the numbers $A(u, v)$ which
  multiplicatively approximate $\Reff(u, v)$ within a factor
  $e^{\beta}$.  Let $v^* = \arg\max_{v \in V} A(u, v)$. By combining
  the inequality \eqref{eq:fixv} with
  \[ \max_{v \in V} \Reff(u, v) \le e^{\beta} \max_{v \in V} A(u, v) = e^{\beta} A(u, v^*) \le e^{2\beta} \Reff(u, v^*), \]
  we obtain $\Reff(u,v^*)\geq \rdim/3$ for some $\beta=\Theta(1)$.
  The algorithm consists of an application of
  Lemma~\ref{lem:singsourceeffres} with $|S| = n$, and a linear scan
  for finding the maximum. Hence, the time bound follows.
\end{proof}

So, Corollary~\ref{cor:smallcut} follows by first using
Lemma~\ref{lem:furthpoints} to find $u',v'$ with $\Reff(u',v') \geq
R/3$, and then apply Theorem~\ref{thm:wtshnim} with the choice of $c$
as described in~(\ref{e:c}).
\end{proof}

\begin{remark}
  We have avoided treating the issues caused by working with an
  approximate potential vector for the sake of clarity. This issue is
  addressed in Appendix \ref{app:accuracy}, Corollary
  \ref{cor:actualsmallcut}.
\end{remark}

\section{Low Effective Resistance Diameter Graph Decomposition}

In this section we prove Theorem \ref{thm:lapsolv}.

\lapsolv*
\begin{proof}
Let $R$ be the target effective resistance diameter and $W$ be the target sum of
the weights of edges that we are going to cut. We will write the
algorithm in terms of $R,W$, and we will optimize for these parameters later in
the proof.
Note that $n=|V|$ is the number of vertices of the original graph $G$,
and it is fixed throughout the execution of the following algorithm.

\vspace*{2mm}

\begin{algorithm}{Effective Resistance Partitioning}{A graph $H$, and parameters $R,W,n$.}{A partition $\Par = \set{V_i \mid i = 1, \ldots, h}$ of $V(H)$.}\label{alg:effrespar2}
  \begin{enumerate}
  \item If there is a vertex $v\in V(H)$ such that $\deg_H(v) \leq W/(2n)$, then delete all the edges incident to $v$. Repeat this step until there are no such vertices in the remaining graph $H$.
  \item Use   Lemma \ref{lem:furthpoints} to find vertices $u,v$ such that $\Reff(u,v)\geq \rdim(H)/3$.
  \item If $\Reff(u, v) \le R$, return $\set{V(H)}$.
  \item Otherwise, find the cut $(U, U^c)$ with
    $\Phi_H(U) \lesssim \frac{(n/W)^\ee}{\sqrt{\ee\cdot R}} \cdot \vol_H(U)^{\ee - 1/2}$ by
    invoking Corollary \ref{cor:smallcut}, with minimum degree at least $W/(2n)$ and $\ee=1/4$.
  \item Call the algorithm recursively on $H[U]$ and $H[U^c]$.
  \item Return the union of the outputs of both recursive calls.
  \end{enumerate}
\end{algorithm}

First of all, by construction, every set $V_i$ in the output partition
satisfies $\rdim(G[V_i])\leq 3R$.  It is not hard to see that the
running time is $\OOs(n\cdot m \cdot \log(w(E)/\min_e w(e)))$, as the
most expensive of the above algorithm takes time $\OOs(m \cdot
\log(w(E)/\min_e w(e)))$, and we make at most $n$ recursive calls.

It remains to calculate the sum of the weights of all edges that we cut.
Note that we cut edges either when a vertex has a low degree or when we find a low conductance set $U$.
We classify the cut edges into two types as follows:
\begin{enumerate}[i)]
 \item Edges $e$ where $e$ is cut as an incident edge of a vertex $v$
   with $\deg_H(v) \leq W/2n$.
 \item The rest of the edges, i.e., edges $e$ where $e\in \partial_H(U)$ for some $U$ where $\Phi_H(U) \lesssim \frac{(n/W)^\ee}{\sqrt{\ee\cdot R}} \vol_H(U)^{\ee-1/2}$.
\end{enumerate}
We observe that we are going to remove edges of type (i) for at most
$n$ times, because each such removal isolates a vertex of $G$.  So, the sum
of the weights of edges of type (i) that we cut is at most $n\cdot
W/2n\leq W/2$. It remains to bound the sum of the weight of edges of
type (ii) that we cut.

We use an amortization argument: Let $\Psi(e)$ stand for the tokens
charged from an edge. We assume that for each edge $e \in E$, the
number of tokens $\Psi(e)$ is initially set to $0$. Every time we make
a cut of type (ii), we assume without loss of generality that
$\vol_H(U) \le \vol(H)/2$ and we modify the number of tokens as
follows
\begin{equation}\label{eq:Psidef} \Psi(e) :=
  \begin{cases}
    \Psi(e) + \frac{w(\partial_H U)}{w(E_H(U))}  & \tif e \in E_H(U)\\
    \Psi(e) & {\rm~ otherwise}.
  \end{cases}
\end{equation}
By definition, after the termination of the algorithm, we have
\begin{equation}
  w\parens*{\text{set of cut edges of type (ii)}} = \sum_{e \in E} \Psi(e) \cdot w(e). \label{eq:cut}
\end{equation}

Therefore, to bound the total weight of type (ii) edges that are cut,
it is enough to show that no edge is charged with too many tokens
provided $R$ is large enough.
\begin{claim}\label{claim:lc2}
  If $R \gtrsim n/(\ee W)$, we will have
  $\Psi(e) \lesssim \frac{4}{\sqrt{R\cdot W/8n}-1}$ for all edges $e
  \in E$ after the termination of the algorithm.
\end{claim}
\begin{proof}
  Fix an edge $e \in E$.  Let $\Diff \Psi(e)$ be the increment of $\Psi(e)$
  due to a cut $(U, U^c)$.  We have
 \begin{equation}\label{eq:Psiinc} \Diff \Psi(e) = \frac{w(\partial_H U)}{w(E_H(U))} = 2 \cdot
      \frac{w(\partial_H U)}{\vol_H(U) - w(\partial_H U)} = 2
      \cdot \frac{1}{\frac{1}{\Phi_H(U)} - 1} \lesssim \frac{2c}{\vol_H(U)^{1/2-\ee}-c},
\end{equation}
where $c$ is chosen as in~(\ref{e:c}) in the proof of Corollary~\ref{cor:smallcut} so that $\Phi(U) \leq c/\vol(U)^{1/2-\ee}$ for the last inequality to hold.
Since the minimum degree is at least $W/2n$ by Step (1) of the algorithm, we have
\[c \asymp
\frac{(2n/W)^\ee}{\sqrt{\ee\cdot R}}.
\]
The minimum degree condition also implies that $\vol_H(U) \geq W/(2n)$.
Note that the denominator of the rightmost term of~(\ref{eq:Psiinc}) is non-negative as long as $\vol_H(U)^{1/2-\eps} \geq (W/2n)^{1/2 - \ee} \geq c$, which holds when $R \gtrsim n/(\ee W)$.

Let $U_0\subseteq V(H_0)$ be the set for which $e$ was charged for the
last time, and in general $U_k\subseteq V(H_k)$ be the $k$-th last set
for which $e$ was charged.  We write $\Delta_k \Psi(e)$ to denote the
increment in $\Psi(e)$ due to $U_k$.

Note that by \eqref{eq:Psidef} we have $e\in E_{H_i}(U_i)$ for all
$i$.  Furthermore, since $\vol_{H_i}(U_i) \leq \vol(H_i)/2 \leq
\vol_{H_{i+1}}(U_{i+1})/2$ for all $i$, we have
\begin{equation}\label{eq:Ukincrement}
\vol_{H_k}(U_k) \geq 2^k \vol_{H_0}(U_0)
\end{equation}
 for all $k\geq 0$. Therefore, using \eqref{eq:Psiinc} and
 \eqref{eq:Ukincrement}, we can write
  \begin{eqnarray*}
    \Psi(e) = \sum_{k \ge 0} \Diff_k \Psi(e) &\le& \sum_{k\geq 0}
    \frac{2c}{\vol_{H_i}(U_i)^{1/2-\ee} -c} \\ &\leq& \sum_{k \ge 0}
    \frac{2c}{(2^{k} \vol_{H_0}(U_0))^{1/2 - \ee} - c} \\ &\le&
    \frac{2c}{\vol_{H_0}(U_0)^{1/2 - \ee}-c} \cdot \sum_{k \ge 0}
    \frac{1}{(2^{1/2 - \ee})^k},
  \end{eqnarray*}
where the last inequality assumes that $\ee < 1/2$.
As argued before, the minimum degree condition implies that every vertex is of degree at least $W/2n$ and thus $\vol_{H_0}(U_0) \geq W/(2n)$.
Therefore, by the geometric sum formula, we have
\[ \Psi(e) \leq \frac{2}{\frac1c (W/2n)^{1/2-\ee}-1} \cdot
\frac1{1-2^{\ee-1/2}}.\]
Plugging the value of $c$ and setting $\ee=1/4<1/2$, we conclude that
\[ \Psi(e) \lesssim \frac{2}{\sqrt{\ee\cdot R \cdot W/2n} -1}
\leq \frac{4}{\sqrt{R\cdot W/8n}-1}.\]
\end{proof}

Setting $R \asymp \delta^2 \cdot n/W$ for a sufficiently large
$\delta^2 >1$ so that the assumption of Claim~\ref{claim:lc2} is
satisfied, it follows from~(\ref{eq:cut}) that the sum of the weights
of all cut edges is at most
$$ W/2 + \sum_e \Psi(e)\cdot w(e) \lesssim W/2 + \frac{w(E)}{\delta}.$$
Setting $W=w(E)/\delta$ proves the theorem.
This completes the proof of Theorem \ref{thm:lapsolv}.
\end{proof}

\section{Conclusions and Open Problems}

We have shown that we can decompose a graph into components of bounded
effective resistance diameter while losing only a small number of edges.
There are few questions which arise naturally from this work.
\begin{enumerate}
\item Can the decomposition in Theorem~\ref{thm:lapsolv} be computed
  in near linear time?  Is this decomposition useful in generating a
  random spanning tree?

\item For the Unique Games Conjecture, Theorem~\ref{thm:lapsolv}
  implies that we can restrict our attention to graphs with bounded
  effective resistance diameter.  Can we solve Unique Games instances
  better in such graphs?  More generally, are there some natural and
  nontrivial problems that can be solved effectively in graphs of
  bounded effective resistance diameter?

\item Is there a generalization of Theorem~\ref{thm:wtshnim} to
  multi-partitioning, i.e.~does the existence of $k$-vertices with
  high pairwise effective resistance distance help us in finding a
  $k$-partitioning of the graph where every cut is very sparse?

\item Theorem~\ref{thm:wtshnim} says that a small-set expander has
  bounded effective resistance diameter.  Is it possible to strengthen
  Theorem~\ref{thm:lapsolv} to show that every graph can be decomposed
  into small-set expanders?  This may be used to show that the
  Small-Set Expansion Conjecture and the Unique Games Conjecture are
  equivalent, depending on the quantitative bounds.
\end{enumerate}

\section*{Acknowledgements}
We would like to thank Hong Zhou for helpful discussions and anonymous
referees for their useful suggestions.

\bibliography{itcsrefs}
\bibliographystyle{alpha}

\appendix

\section{Robustness of the Proof of Theorem \ref{thm:wtshnim}}\label{app:accuracy}

We avoided the issue of picking the accuracy parameter $\eps > 0$
for the Laplacian solver we used in Corollary \ref{cor:smallcut}.
Here, we want to show that the proof is robust enough to small
perturbations in the potential vector, i.e.~using a
Laplacian solver to estimate $s$-$t$ potential vector $\vecc p \in \RR^V$
by the vector $\hat{\vecc p}$, additively within an accuracy of
$\eta$, we can still recover our sparse cut.

We first start by noting that $|p(v) - \hat p(v)| \le \eta$ is implied by the stronger inequality,
\begin{equation}
  \norm{\hat{\vecc p} - \vecc p}^2 \le \eta^2. \label{eq:desiredacc}
\end{equation}

We will show that if $\eta$ is polynomially small in the input data,
we can still find a sparse cut. Our plan is as follows.
\begin{itemize}
\item We will figure out how small we should set the Laplacian solver accuracy
  $\eps$ to ensure \eqref{eq:desiredacc} (Lemma \ref{lem:solveracc}).
\item We will show that using the mild-expansion of the threshold sets $T_{\hat p}$
  of the vector $\hat{\vecc p}$, we can still prove upper bounds on
  the effective resistance (Theorem \ref{thm:actualwtshnim}).
\item Analogously to Corollary \ref{cor:smallcut}, we will show that
  by way of contraposition the existence of a pair with large
  effective resistance distance means one of the threshold sets
  $T_{\hat p}$ does not satisfy the mild expansion property (Corollary
  \ref{cor:actualsmallcut}).
\end{itemize}

\subsection{Eigenvalue Bound}

We start with a simple eigenvalue bound that will be used to bound the accurancy needed.

\begin{claim}
  For any connected weighted graph $G = (V, E, \vecc w)$, we have
  \begin{equation}
 \lambda_2(G) \gtrsim \min_e w(e) \cdot \parens*{\frac{\min_e w(e)}{w(E)} }^2. \label{eq:eigbd}
\end{equation}
\end{claim}

\begin{proof}
  For any connected weighted graph $G = (V, E, \vecc w)$, we have the following conductance bound,
  \[ \min\set*{\frac{|\partial(S)|}{\vol(S)} : \vol(S) \le \vol(G)/2} = \Phi(G) \ge \frac{\min_e w(e)}{w(E)},\]
which implies
\[ \frac{\min_e w(e)}{w(E)} \lesssim \sqrt{\tilde \lambda_2(G)}
\quad \Longleftrightarrow \quad
\tilde \lambda_2(G) \gtrsim  \parens*{\frac{\min_e w(e)}{w(E)}}^2\]
by \ref{eq:cheeg}.
Note that
\begin{equation}
  \lambda_2(G) = \min_{\vecc u \perp \vecc 1} \frac{\langle \vecc u, \matr L \vecc u \rangle}{\langle \vecc u, \vecc u\rangle} = \min_{\vecc v \perp {\matr D^{1/2} \vecc 1}} \frac{\langle \matr D^{-1/2} \vecc v, \matr L \matr D^{-1/2} \vecc v \rangle}{\langle \matr D^{-1/2} \vecc v, \matr D^{-1/2} \vecc v \rangle}
\end{equation}
where the last equality follows by a change of variables $\vecc u = \matr D^{-1/2} \vecc v$. This implies that
\begin{equation}
  \lambda_2(G) = \min_{\vecc v \perp \matr D^{1/2}\vecc 1} \frac{\langle \vecc v, \tilde{\matr L} \vecc v \rangle}{\langle\vecc v, \matr D^{-1} \vecc v\rangle} \ge \min_{\vecc v \perp \matr D^{1/2} \vecc 1} \frac{\langle \vecc v, \tilde{\matr L} \vecc v \rangle}{\norm{\matr D^{-1}} \cdot \langle \vecc v, \vecc v\rangle} \ge \frac{1}{\norm{\matr D^{-1}}} \cdot \tilde{\lambda}_2(G).
\end{equation}
Using $\norm{\matr D^{-1}}^{-1} = \min_v \deg(v) \ge \min_e w(e)$,
we obtain $\lambda_2(G) \ge \min_e w(e) \cdot \tilde \lambda_2(G)$. Combining everything, we get
\begin{equation}
 \lambda_2(G) \gtrsim \min_e w(e) \cdot \parens*{\frac{\min_e w(e)}{w(E)} }^2
\end{equation}
hence proving the claim.
\end{proof}

\subsection{Picking the Laplacian Solver Accuracy}

For $\vecc b = \vecc e_s - \vecc e_t$, the Spielman Teng Solver in
Lemma~\ref{lem:stsolver} produces a vector $\hat{\vecc p}$ such that
\[
\norm{\hat{\vecc p} - \matr L^\dagger \vecc b}_{\matr L} \le
\zeta \cdot \norm{\matr L^\dagger \vecc b}_L.
\]
Letting $\vecc p$ be the $s$-$t$ electric potential vector, this becomes
\[ \norm{\hat{\vecc p} - \vecc p}_{\matr L} \le \zeta \cdot
  \norm{\vecc p}_{\matr L}.\]
Using the definition of the $\matr L$-norm, we have
\[ \inpr*{\hat{\vecc p} - \vecc p, \matr L (\hat{\vecc p} - \vecc p) } \le \zeta^2 \langle \vecc p, \matr L \vecc p\rangle = \zeta^2 \cdot \Reff(s,t).\]
Since we are working on the space orthogonal to the nullspace, both
$\hat{\vecc p}, \vecc p \perp \vecc 1$ and thus
$(\hat{\vecc p} - \vecc p) \perp \vecc 1$.
It follows from the definition of $\lambda_2(G)$ that
\[ \lambda_2(G) \cdot \norm{\hat{\vecc p} - \vecc p}^2 \le \zeta^2 \cdot \Reff(s, t).\]
By the eigenvalue bound in \eqref{eq:eigbd}, we have
\[ \norm{\hat{\vecc p} - \vecc p}^2 \lesssim \zeta^2 \cdot \Reff(s,t) \cdot \frac{w(E)^2}{(\min_e w(e))^3}.\]
Using the trivial bound $\Reff(s, t) \le \frac{m}{\min_e w(e)}$ in a connected graph, we get
\[  \norm{\hat{\vecc p} - \vecc p}^2 \lesssim \zeta^2 \cdot m \cdot \frac{w(E)^2}{(\min_e w(e))^4}.\]
Therefore, we can set
\[\zeta \asymp \frac{\eta \cdot (\min_e w(e))^2}{w(E) \cdot \sqrt{m}}\]
to get the desired accuracy in \eqref{eq:desiredacc}.
The above argument is summarized in the following lemma.

\begin{lemma}\label{lem:solveracc}
  Given a connected weighted graph $G = (V, E, \vecc w)$,
  it is possible to compute an estimate $\hat{\vecc p}$ of the
  $s$-$t$ electric potential vector
  $\vecc p$ within an additive accuracy of $\eta$ using the
  Spielman-Teng solver with accuracy
  \[\zeta \asymp \frac{\eta \cdot (\min_e w(e))^2}{w(E) \cdot \sqrt{m}}.\]
\end{lemma}

\subsection{How Small Should We Pick $\eta$?}

In the proof of Theorem \ref{thm:wtshnim}, we used the actual
potential vector to bound the effective resistance. This is too
expensive for algorithmic purposes. We now show that we can use the
estimate $\hat{\vecc p}$ and the potential sets
\[ T_{\hat p} = \set*{ v \in V : \hat{p}(v) \ge \hat p }\]
at a small cost.
We will show that the mild-expansion of these cuts allows us to bound
the effective resistance from above, just as in Theorem~\ref{thm:wtshnim}.

\begin{theorem}\label{thm:actualwtshnim}
  Let $\hat{\vecc p}$ be an additive $\eta$-approximation of the
  electric potential vector $\vecc p$ between $s$ and $t$, i.e.~
  \[ |\hat p(u) - p(u)| \le \eta \quad \forall u \in V. \]
  If all the threshold cuts $T_{\hat p}$ satisfy the mild expansion condition,
  \[ w(\partial T_{\hat p}) \ge c \cdot \vol(T_{\hat p})^{1/2 + \ee} \]
  Then, we have
  \[ \Reff(s, t) \lesssim \frac{1}{\ee c^2 \deg(s)^{2\ee}} + \frac{1}{\ee c^2 \deg(t)^{2\ee}} + \frac{\eta (48 m^{1/2 - \ee} \log n + 2c)}{c}. \]
\end{theorem}

\begin{proof}
  The proof will be very similar to that of Theorem~\ref{thm:wtshnim},
  we will just highlight the differences and carry out the relevant
  computations.

To follow the proof of Theorem~\ref{thm:wtshnim}, we need an upper bound on the quantity
\begin{equation}
    \sum_{e \in \partial T_{\hat p}} w(e) \cdot \Abs{\Diff_e \hat{\vecc p}} \le \sum_{e \in \partial T_{\hat p}} w(e) \cdot (\Abs{\Diff_e \vecc p} + 2\eta) =  \sum_{e \in \partial T_{\hat p}} w(e) \cdot \Abs{\Diff_e \vecc p} + 2\eta \cdot w(\partial T_{\hat p}). \label{eq:ubound}
\end{equation}
where for the first inequality we have used the triangle inequality
\[  \Abs{\Diff_e \hat{\vecc p}} = |\hat p(e^+) - \hat p(e^-)| \le |p(e^+) -
  p(e^-)| + |p(e^+) - \hat p(e^+)| + |p(e^-) - \hat p(e^-)| \le
  \Abs{\Diff_e \vecc p} + 2\eta.\]
Bounding the RHS of \eqref{eq:ubound} is certainly possible by bounding
\begin{equation}
  \sum_{e \in \partial T_{\hat p}} w(e) \cdot \Abs{\Diff_e \hat{\vecc p}}.
\end{equation}
In Theorem~\ref{thm:wtshnim}, we used the Equation \eqref{eq:conserv1} to bound the analogous term, i.e.~
\[\sum_{e \in \partial S_p} |f(e)| = \sum_{e \in \partial S_p} w(e) \cdot \Abs{\Diff_e \vecc p} =  1.\]
We were allowed to do this because $S_p$ is a threshold set, i.e.~the $st$
electric flow flows in one direction: from $S_p$ to $S_p^c$. This
means that flow conservation insures
$\sum_{e \in \partial(S_p)} |f(e)| = 1$, as $s$ has a flow deficit of
a unit, and $t$ has a flow surplus of a unit. This is no longer true
for $T_{\hat p}$ since $T_{\hat p}$ is no longer a threshold set of
the true potential vector $\vecc p$, i.e.~we do not necessarily have
$\sum_{e \in \partial(T_{\hat p})} |f(e)| = 1$.

Before we go on further, we adopt the following convention of taking $uv \in
\partial T_{\hat p}$ to be an edge with $u \in T_{\hat p}$ and $v \not\in
T_{\hat p}$.

In our case the conservation of flow still implies,
\begin{equation}
  \sum_{e \in \partial T_{\hat p}} f(e) = 1. \label{eq:conserv}
\end{equation}
We will take $P^{-}$ to be the set of edges
$uv \in \partial(T_{\hat p})$ with $p(u) < p(v)$, and $P^+$ to be the
set of edges $uv \in \partial(T_{\hat p})$ with $p(u) \ge p(v)$.

Now, note that \eqref{eq:conserv} rewrites into,
\[ \sum_{uv \in P^+} w(uv) \cdot \Abs{\Diff_{uv} \vecc p} - \sum_{uv \in P^-} w(uv) \cdot \Abs{\Diff_{uv} \vecc p} = 1.\]
In particular, we can manipulate this to obtain
\begin{equation}
  \sum_{uv \in \partial T_{\hat p}} w(uv) \cdot \Abs{\Diff_{uv} \vecc p} = \sum_{uv \in P^-} w(u, v) \cdot \Abs{\Diff_{uv} \vecc p} + \sum_{uv \in P^-} w(u, v) \cdot \Abs{\Diff_{uv} \vecc p} =
  1 + 2  \sum_{uv \in P^-} w(u, v) \cdot \Abs{\Diff_{uv} \vecc p}.\label{eq:absflow}
\end{equation}
We see now, proving an upper bound on \eqref{eq:ubound} boils down to upper
bounding $\Abs{\Diff_{uv} \vecc p}$ for $uv \in P^-$. This can be done
by noting $\hat p(u) > \hat p(v)$ (as $uv \in \partial(T_{\hat p})$)
and $\vecc p \approx \hat{\vecc p}$. Formally, we have
\[  p(u) + \eta \ge \hat p(u) \ge \hat p > \hat p(v) \ge p(v) - \eta,\]
which means, we must either have $p(u) \ge p(v)$ or it must be the case that we have $p(u) < p(v)$ and $\Abs{\Diff_{uv} \vecc p} \le
2\eta$. This readily implies the following inequality,
\begin{equation}
  \sum_{e \in P^-} w(e) \cdot \Abs{\Diff_e \vecc p} \le 2\eta \cdot w(\partial T_p).
\end{equation}
Plugging this in \eqref{eq:absflow}, we get
\[  \sum_{e \in \partial T_{\hat p}} w(e) \cdot \Abs{\Diff_e \vecc p} = \sum_{e \in P^+} w(e) \cdot \Abs{\Diff_e \vecc p} + \sum_{e \in P^-} w(e) \cdot \Abs{\Diff_e \vecc p} \le 1 + 4\eta \cdot w(\partial T_{\hat p}).\]
Combining this with \eqref{eq:ubound} yields,
\begin{equation}
  \sum_{e \in \partial T_{\hat p}} w(e) \cdot \Abs{\Diff_e \hat{\vecc p}} \le 1 + 6 \eta \cdot w(\partial T_{\hat p}). \label{eq:ubound2}
\end{equation}
With this bound, we can proceed as in the proof of Theorem \ref{thm:wtshnim}.
Normalizing \eqref{eq:ubound2} we obtain,
\[  \sum_{e \in \partial T_{\hat p}} \frac{w(e)}{w(\partial T_{\hat p})} \cdot \Abs{\Diff_e \hat{\vecc p}} \le \frac{1}{w(\partial T_{\hat p})} + 6\eta. \]

Since $\mu(e) = w(e)/w(\partial T_{\hat p})$ is a probability
distribution, analogously as in the proof of Theorem
\ref{thm:wtshnim}, we obtain a set $F \subseteq \partial T_{\hat p}$
by Markov's inequality such that
\begin{itemize}
\item all edges $f \in F$ satisfy $\Abs{\Diff_f \hat{\vecc p}} \le \frac{2}{w(\partial T_{\hat p})} + 12\eta$
\item the set $F$ is ``large'', i.e.~ $w(F) = \mu(F) \cdot w(\partial T_{\hat p}) \ge \frac{1}{2} \cdot w(\partial T_{\hat p})$.
\end{itemize}

Now, analogously as in the proof of Theorem \ref{thm:wtshnim}, we get
\[  \vol(T_{\hat p - \frac{2}{w(\partial T_{\hat p})} - 12\eta}) \ge \vol(T_{\hat p}) + \frac{1}{2} \cdot w(\partial T_{\hat p}). \]
Using \ref{eq:wbdbd} of $T_{\hat p}$, this implies
\[  \vol\parens*{T_{\hat p - \frac{2}{c \vol(T_{\hat p})^{1/2 + \ee}} - 12\eta }} \ge \vol(T_{\hat p}) + \frac{c \vol(T_{\hat p})^{1/2 + \ee} }{2}.\]
Iterating this for $(\frac{2}{c}\vol(T_{\hat p})^{1/2 - \ee})$-times, we obtain
\begin{equation}
  \vol\parens*{T_{\hat p - \frac{2}{c^2 \vol(T_{\hat p})^{2\ee}} - \frac{24 \eta \cdot \vol(T_{\hat p})^{1/2 - \ee}}{c}}} \ge 2 \vol(T_{\hat p}).
\label{eq:tdouble}
\end{equation}
Now, similarly we set $\hat p_0 = \hat p(s)$ and inductively extend
this to $\hat p_k$ for $k > 0$ by
\[  \hat p_{k + 1} = \hat p_k - \frac{2}{c^2 \vol(T_{\hat p})^{2\ee}} - \frac{24 \eta \cdot \vol(T_{\hat p})^{1/2 - \ee}}{c}.\]
Since $s \in T_{\hat p_0}$, we know $\vol(T_{\hat p_0}) \ge \deg(s)$,
and by \eqref{eq:tdouble} $\vol(T_{\hat p_k}) \ge 2^k \deg(s)$.
There exists some $k^\star \lesssim \log n$ satisfying,
\begin{equation}
  \vol(G) \ge 2 \vol(T_{\hat p_{k^\star}}) \ge \vol(G)/2.
\end{equation}
It follows that
\[  \hat p_{k^\star} - \hat p_0 = \sum_{j = 0}^{k^\star} \frac{2}{c^2
    \vol(T_{\hat p})^{2\ee}} + \sum_{j = 0}^{k^\star} \frac{24 \eta \cdot \vol(T_{\hat p})^{1/2 - \ee}}{c}.\]
Using the bound $\vol(T_{\hat p_k}) \ge 2^k \deg(s)$ and the geometric sum formula, we have
\[  \hat p_{k^\star} - \hat p_0 \lesssim \sum_{j = 0}^{k^\star}
  \frac{2}{c^2 2^{j2\ee} \deg(s)^{2\ee}} + \sum_{j = 0}^{k^\star}
  \frac{24 \eta \cdot \vol(T_{\hat p})^{1/2 - \ee}}{c} \le
  \frac{1}{\ee c^2 \deg(s)^{2\ee}} + \sum_{j = 0}^{k^\star} \frac{24
    \eta \cdot \vol(T_{\hat p})^{1/2 - \ee}}{c}.\]
Using the naive bounds $\vol(T_{\hat p}) \le m$ and $k^\star \lesssim \log n$,
we obtain
\[  \hat p_{k^\star} - \hat p_0 \lesssim \frac{1}{\ee c^2 \deg(s)^{2\ee}} + \frac{24 \eta \cdot m^{1/2 - \ee}\log n}{c}.\]
Using similar arguments (sending flow from $t$ to $s$), we see that
more than half the volume of the vertices has $\hat{\vecc p}$
potential difference at least
\[ \frac{1}{\ee c^2 \deg(t)^{2\ee}} + \frac{24 \eta \cdot m^{1/2 -
      \ee}\log n}{c}\]
to the vertex $t$. Therefore, we can prove the following potential
difference upper-bound with respect to $\hat{\vecc p}$,
\[  \hat p(s) - \hat p(t) \le \frac{1}{\ee c^2 \deg(s)^{2\ee}} + \frac{1}{\ee c^2 \deg(t)^{2 \ee}} + \frac{48 \eta \cdot m^{1/2 - \ee} \log n}{c}.\]
Using the triangle inequality,
\[ p(s) - p(t) \le \hat p(s) - \hat p(t) + |\hat p(s) - p(s) | + |\hat p(t) - p(t) | \le \hat p(s) - \hat p(t) + 2 \eta.\]
We can conclude that
\begin{eqnarray*}
  p(s) - p(t) & \le & \frac{1}{\ee c^2 \deg(s)^{2\ee}} + \frac{1}{\ee
    c^2 \deg(t)^{2 \ee}} + \frac{48 \eta \cdot m^{1/2 - \ee} \log n}{c} +
  2\eta \\ & = & \frac{1}{\ee c^2 \deg(s)^{2\ee}} + \frac{1}{\ee c^2
    \deg(t)^{2 \ee}} + \frac{\eta (48 m^{1/2 - \ee} \log n + 2c)}{c}.
\end{eqnarray*}
Noting that $p(s) - p(t) = \Reff(s, t)$ completes the proof.
\end{proof}

We use the above result to complete the proof of our algorithmic result, in which we will choose $\eta$.

\begin{corollary}\label{cor:actualsmallcut}
  Let $G = (V, E, \vecc w)$ be a connected weighted graph. If $\deg(v)\geq
  1/\alpha$ for all $v\in V$, then for any $0<\eps<1/2$, there is a
  cut $(U, U^c)$ such that
  \[\Phi(U)  \lesssim \frac{\alpha^\ee}{\sqrt{\rdim\cdot \ee}} \cdot \vol(U)^{\ee - 1/2}.
  \]
  Further, the set $U$ can be found in time $\OOs\parens*{m \cdot
    \log\parens*{\frac{w(E)}{\min_e w(e)}}}$.
\end{corollary}

\begin{proof}
  The proof will be the same as that of Corollary
  \ref{cor:smallcut}. For Corollary \ref{cor:smallcut}, we used
  Theorem \ref{thm:wtshnim} to get our non-expanding cut, here we will
  use Theorem \ref{thm:actualwtshnim}. By Lemma \ref{lem:furthpoints},
  we can compute vertices $u, v \in V$ which satisfies,
  \[ \Reff(u, v) \ge \rdim/3 \]
  in time $\OOs(m)$. As in Corollary \ref{cor:smallcut}, we pick
  \[ c \asymp \sqrt{\frac{ \frac{1}{\deg(u)^{2\ee}} + \frac{1}{\deg(v)^{2\ee}} }{\rdim \cdot \ee }}. \]
  To contradict the bound from Theorem \ref{thm:actualwtshnim} asymptotically, we need to pick $\eta$ to satisfy
  \[ \frac{1}{\ee c^2} \cdot \parens*{\frac{1}{\deg(u)^{2\ee}} + \frac{1}{\deg(v)^{2\ee}}} \geq \frac{\eta \cdot (48 \cdot m^{1/2 - \ee} \log n + 2c)}{c}.\]
  This will allow us to ignore the additional error term we get in the
  proof of Theorem \ref{thm:actualwtshnim}, by settling for a bigger
  constant that will be hidden in the $\lesssim$-notation.

  By the AM-GM inequality, the choice
  \[ \eta := \frac{1}{\ee \cdot c} \cdot \parens*{\frac{1}{\deg(u)^{2\ee}} + \frac{1}{\deg(v)^{2\ee}}} \cdot \frac{1}{\sqrt{96 \cdot m^{1/2} \log n}}\]
  certainly satisfies this.
  Note that
  \begin{eqnarray*}
    \eta & = & \frac{1}{\ee \cdot c} \cdot
    \parens*{\frac{1}{\deg(u)^{2\ee}} + \frac{1}{\deg(v)^{2\ee}}}
    \cdot \frac{1}{\sqrt{96 \cdot m^{1/2} \log n}}\\
    & \gtrsim & \frac{1}{\ee \cdot c^{3/2}} \cdot \frac{1}{w(E)^{2 \ee}}
      \cdot \frac{1}{\sqrt{m^{1/2} \log n}}\\
    & \asymp & \frac{1}{\ee} \cdot \parens*{\frac{\rdim \cdot
      \ee}{\frac{1}{\deg(u)^{2\ee}} + \frac{1}{\deg(v)^{2\ee}}}}^{3/4} \cdot
    \frac{1}{w(E)^{2 \ee}} \cdot \frac{1}{\sqrt{m^{1/2} \log n}}\\
    & \gtrsim & \frac{1}{\ee} \cdot \parens*{\rdim \cdot \ee \cdot (\min_e w(e))^{2 \ee}}^{3/4} \cdot
    \frac{1}{w(E)^{2 \ee}} \cdot \frac{1}{\sqrt{m^{1/2} \log n}}\\
    & \geq & {\rdim}^{\frac{3}{4}} \cdot \frac{(\min_e w(e))^{3\ee/2}}{w(E)^{2 \ee}} \cdot \frac{1}{\sqrt{m^{1/2} \log n}}.
  \end{eqnarray*}

  Since the smallest possible $\rdim$ is when all the edges act as
  parallel resistors between two vertices, we have
  \[ \rdim \ge \frac{1}{\sum_{e \in E} w(e)} = \frac{1}{w(E)}. \]
  Plugging this into the above computation, we get
  \begin{equation}
    \eta \gtrsim \frac{(\min_e w(e))^{3\ee/2}}{w(E)^{3/4+2 \ee}} \cdot \frac{1}{\sqrt{m^{1/2} \log n}}.
  \end{equation}

  Using this with Lemma \ref{lem:solveracc}, we see that we need to pick
  \begin{equation}
    \zeta \asymp \frac{\eta \cdot (\min_e w(e))^2}{w(E) \cdot \sqrt{m}} \gtrsim \frac{(\min_e w(e))^{2+3\ee/2}}{w(E)^{7/4+2\ee}} \cdot \frac{1}{\sqrt{m^{3/2} \log n}}.
  \end{equation}
  as the accuracy for the Spielman-Teng solver.

  In particular, this means that the Laplacian solver will take time
  \[ \OOs(m \cdot \log(1/\zeta)) = \OOs\parens*{m \cdot \plog(m) + m \cdot \log\parens*{\frac{w(E)}{\min_e w(e)}}}  = \OOs\parens*{m \cdot \log\parens*{\frac{w(E)}{\min_e w(e)}}}.\]
  The rest of the proof follows as in Corollary \ref{cor:smallcut}.
\end{proof}

\end{document}